\documentclass[letter,UKenglish,autoref,11pt]{article}

\usepackage{microtype}

\usepackage{soul}
\usepackage{lipsum,graphicx}
\usepackage{caption}
\usepackage[dvipsnames,svgnames,x11names,table]{xcolor}

\usepackage{array}
\usepackage{amsmath}
\usepackage{amssymb}
\usepackage{amsthm}
\usepackage{fullpage}
\usepackage{complexity}

\usepackage{todonotes}

\usepackage{libertine}
\usepackage{libertinust1math}

\usepackage[T1]{fontenc}

\usepackage[most]{tcolorbox}
\usepackage{fontawesome}
\usepackage{inconsolata}
\definecolor{bloodRed}{RGB}{92,3,37}
\definecolor{slightPurple}{RGB}{166,6,66}
\definecolor{slighterPurple}{RGB}{184,43,87}
\definecolor{elegantblue}{RGB}{100,140,180}
\definecolor{elegantcoral}{RGB}{220,140,120}
\definecolor{elegantsage}{RGB}{140,170,140}
\definecolor{elegantlavender}{RGB}{180,150,180}
\definecolor{elegantgold}{RGB}{220,180,100}
\definecolor{elegantgray}{RGB}{120,120,120}

\usepackage[colorlinks,backref,pagebackref]{hyperref}
\hypersetup{
  linkcolor=slightPurple, 
  citecolor=slighterPurple, 
  urlcolor= bloodRed 
}

\allowdisplaybreaks
\usepackage{multicol}
\usepackage{subcaption}

\usepackage[ruled,vlined, linesnumbered]{algorithm2e}
\usepackage[capitalize, nameinlink, noabbrev]{cleveref}

\usepackage{tikz}

\usetikzlibrary{chains,fit}
\usetikzlibrary{shapes.geometric}
\usetikzlibrary{matrix,positioning,calc}
\usetikzlibrary{decorations.markings}
\usetikzlibrary{decorations.pathmorphing}
\usetikzlibrary{decorations.pathreplacing}

\tikzstyle{vertex}=[circle, draw, thick, inner sep=2pt, minimum width=4pt]

\tikzstyle{vertbox}=[draw, inner sep=0pt, minimum size=8pt]

\usepackage[dvipsnames]{xcolor}
\usepackage{subfiles}
\usepackage[normalem]{ulem}

\usepackage{thmtools}
\usepackage{thm-restate}

\crefname{algocf}{alg.}{algs.}
\Crefname{algocf}{Algorithm}{Algorithms}

\newtheorem{theorem}{Theorem}

\newtheorem{observation}[theorem]{Observation}

\newtheorem{definition}[theorem]{Definition}
\newtheorem{claim}[theorem]{Claim}

\newcommand{\PSP}{\mathrm{MSPC}}

\title{Shortest Paths with Linear Edge Weights}

\author{\vspace{0.1cm}\\
    Suryajith Chillara\\
    CSTAR, IIIT Hyderabad, India\\
    \texttt{suryajith.chillara@iiit.ac.in}
    \and
    \vspace{0.1cm}\\
    Kshitij Gajjar\\
    CSTAR, IIIT Hyderabad, India\\
    \texttt{kshitij@iiit.ac.in}
    \and
    Nithish Raja\thanks{Part of this work was done while the author was a Master's student at CSTAR, IIIT Hyderabad, India.}\\
    TU Eindhoven, Netherlands\\
    \texttt{n.r.raja@tue.nl}
    \vspace{0.5cm}
}

\begin{document}

\maketitle

\begin{abstract}
    
    We study shortest paths in directed graphs whose edge weights are of the form $$ \mathsf{wt}(e) = a_{e,1} \lambda_1 + a_{e,2} \lambda_2 + a_{e,3} \lambda_3 + \cdots + a_{e,d} \lambda_d + a_{e,d+1}.$$
    Here, each $a_{e,i}\in\mathbb{R}$ is a fixed constant for each edge $e$, whereas each $\lambda_i\in\mathbb{R}$ is common across the entire graph. So, there could be different shortest paths in the graph for different values of the $\lambda_i$'s. The number of such shortest paths is of interest in several combinatorial optimization problems. This is called the \emph{Parametric Shortest Paths} problem, and has been studied since the 1980s.
    
    For $d=1$, Carstensen (1983) showed that the number of shortest paths in $n$-vertex graphs is at most $n^{O(\log n)}$. She also proved a matching lower bound of $n^{\Omega(\log n)}$, which was later refined by Mulmuley \& Shah (2001). For $d=2$, Gajjar \& Radhakrishnan (2019) showed an upper bound of $n^{O(\log^2 n)}$. Barth, Funke \& Proissl (2022) generalized their result to prove an upper bound of $n^{O_d(\log^d n)}$ for all positive integers $d$. The lower bound did not undergo any improvement over the years.

    In this paper, we close this long line of research by showing an $n^{O(d\log n)}$ upper bound for all positive integers $d$, exponentially improving the previous upper bound. We observe that a matching lower bound of $n^{\Omega(d\log n)}$ can be obtained by trivially extending existing lower bound constructions for $d=1$. We also show that our proof can be adapted to work for undirected graphs with positive edge weights. Furthermore, for directed graphs whose edge weights are univariate polynomials of degree at most $q$, we prove an upper bound of $n^{O(\log{n}+\log{q})}$.
    
    Finally, building upon work on the \emph{Point Location} problem by Ezra, Har-Peled, Kaplan \& Sharir (2020), we construct a \emph{Shortest Path Oracle} which takes as input a point $\overline{x}\in \mathbb{R}^d$, and outputs a shortest path at $\overline{\lambda}=\overline{x}$ in sublinear time (for a wide regime of $d$).
    
    All earlier upper bound proofs proceeded by arranging the vertices of the graph in layers, splitting the graph across its middle layer into two ``halves'', and then recursing on each half-graph. We deviate from this proof methodology by ``halving'' the graph in a different way: we eliminate all the odd-numbered layers and retain only the even-numbered layers, whilst maintaining requisite shortest paths of the original graph. We then view shortest paths in the half-graph as convex objects in $d$-dimensional space, which leads us to the required recurrence.
\end{abstract}


\tableofcontents

\newpage

\section{Introduction}\label{sec:introduction}

Computing shortest paths in graphs is a fundamental problem in computer science, dating back to over seven decades. Some of the most popular algorithms to efficiently find a shortest path in a given network, namely Dijkstra's algorithm~\cite{Dijkstra1959ANO}, the Bellman-Ford algorithm~\cite{bellman1958routing,ford1956network,moore1959shortest}, the Floyd-Warshall algorithm~\cite{floyd1962algorithm,warshall1962theorem,roy1959transitivite}, and Johnson's algorithm~\cite{johnson1977efficient} are all now taught in basic Algorithms courses worldwide.

In fact, the most popular algorithm used in practice is the $A^*$ algorithm~\cite{hart1968astar} to compute shortest paths (which is basically Dijkstra's algorithm with a heuristic), even more popular than the famed Fast-Fourier Transform~\cite{CooleyTukey1965, Heidman_Johnson_Burrus_1985} that performs multiplications!

Most of the research on shortest paths is focused on graphs whose edge weights are fixed, including all the algorithms mentioned above. In most practical scenarios, however, the edge weights are changing in real time. Some common examples are the traffic on a street in a road network, the network traffic on a LAN wire connecting two nodes in a computer network, the strength of the signal being passed across two cellular masts (cellphone towers) in a telecommunication network. 
Such situations necessitate the study of graphs in which the edge weights are varying with time. There are mainly four lines of work in this regard:

\begin{enumerate}
    \item[(i)] Edge weights are varying  probabilistically (generally with the graph having a probability distribution, and each edge having a mean and a variance~\cite{NKBM_2006,frieze1985shortest}).
    
    \item[(ii)] Weighted edges are coming in an online fashion (as in temporal or dynamic graphs
    
    \item[(iii)] Edge weights are given by time-varying functions (each edge is assigned a function of time, mapping it to a real number denoting the weight of the edge~\cite{Carstensen_1983,erickson2010maximum}).
    
    \item[(iv)] Edge weights are given by time-varying functions (but path weights are functional compositions of their constituent edge weight functions~\cite{FoschiniHS2014,Gajjar_Varma_Chattarjee_Radhakrishnan_21}).
\end{enumerate}

All four of these aspects have been studied in the past. 
However, throughout this paper, we will be focusing only on the third aspect (item (iii) above), known as parametric shortest paths.

\subsection{Parametric Shortest Paths}

Research on parametric shortest paths began in the 1980s, when Carstensen~\cite{Carstensen_1983} studied the minimum size of a shortest path cover (MSPC) of graphs whose edge weights are linear functions of a parameter $\overline{\lambda}$.

\begin{definition}[Minimum Shortest Path Cover (MSPC)]\label{def:mspc}
Let $G$ be a directed graph (with two special vertices $s$ and $t$), whose edge weights are linear functions of $d$ variables $\lambda_1,\lambda_2,\ldots,\lambda_d$; that is, $$\mathsf{wt}(e)=a_{e,1} \lambda_1 + a_{e,2} \lambda_2 + \cdots + a_{e,d} \lambda_d + a_{e,d+1}$$ for each edge $e\in E(G)$, and $a_{e,i}\in\mathbb{R}$ for each $i\in\{1,2,\ldots,d+1\}$. (See \Cref{fig:planar3x3} for an example of such a graph.) The weight of a path $P$ from $s$ to $t$ is defined in the usual way as the sum of the weights of its constituent edges: $$\mathsf{wt}(P)=\sum_{e\in P}\mathsf{wt}(e)=\sum_{e\in P} \left(a_{e,1} \lambda_1 + a_{e,2} \lambda_2 + \cdots + a_{e,d} \lambda_d + a_{e,d+1}\right).$$ Let $\mathcal{P}_{s,t}$ be the set of paths from $s$ to $t$ in $G$. Then, $\mathcal{S}$ is a Minimum Shortest Path Cover ($\PSP$) of $G$ if $\mathcal{S}$ is a minimum-sized subset of $\mathcal{P}_{s,t}$ such that for every $\overline{x}\in\mathbb{R}^d$, there is a path $P\in\mathcal{S}$ which is a shortest path in $G$ at $\overline{\lambda}=\overline{x}$; that is, the path $P$ is a shortest path from $s$ to $t$ in the fixed-edge weight graph obtained upon setting $(\lambda_1,\lambda_2,\ldots,\lambda_d) = (x_1,x_2,\ldots,x_d)$ in all the edge weights of $G$. (See \Cref{fig:allpathsmcsp} for an example of an MSPC.)
\end{definition}

	\begin{figure}
		\centering
		\begin{tikzpicture}[scale=1.4]
			
			\begin{scope} [xshift=-3.25cm]
			{
                \node[vertex] (v00) at (0,0) [fill=blue,circle, draw, inner sep=0pt, minimum size=8pt, label=below left:{\Large $s$}] {};
                

                \node[vertex] (v02) at (0,2) [fill=blue,circle, draw, inner sep=0pt, minimum size=8pt] {};
                

                \node[vertex] (v04) at (0,4) [fill=blue,circle, draw, inner sep=0pt, minimum size=8pt] {};
                

                \node[vertex] (v20) at (2,0) [fill=blue,circle, draw, inner sep=0pt, minimum size=8pt] {};


                \node[vertex] (v22) at (2,2) [fill=blue,circle, draw, inner sep=0pt, minimum size=8pt] {};


                \node[vertex] (v24) at (2,4) [fill=blue,circle, draw, inner sep=0pt, minimum size=8pt] {};

				
				\node[vertex] (v40) at (4,0) [fill=blue,circle, draw, inner sep=0pt, minimum size=8pt] {};


                \node[vertex] (v42) at (4,2) [fill=blue,circle, draw, inner sep=0pt, minimum size=8pt] {};


                \node[vertex] (v44) at (4,4) [fill=blue,circle, draw, inner sep=0pt, minimum size=8pt, label=above right:{\Large $t$}] {};

			}
			\end{scope}
			
			\begin{scope} [decoration={markings, mark=at position 0.6 with {\arrow[scale=2,>=stealth,gray]{>}}}]
				\draw [postaction={decorate}] (v00)-- node[below, opacity=0.9] {\footnotesize{$\lambda_1 + 3\lambda_2 - 4\lambda_3 - 7$}} (v20);
				\draw [postaction={decorate}] (v20)-- node[below, opacity=0.9] {\footnotesize{$-15\lambda_2 + 6\lambda_3 + 2$}}(v40);
				\draw [postaction={decorate}] (v02)-- node[above, opacity=0.9] {\footnotesize{$3\lambda_1 + \lambda_2 - 4\lambda_3 - 9$}}(v22);
				\draw [postaction={decorate}] (v22)-- node[above, opacity=0.9] {\footnotesize{$\lambda_1 - 9\lambda_2 + 4\lambda_3 + 7$}}(v42);
				\draw [postaction={decorate}] (v04)-- node[above, opacity=0.9] {\footnotesize{$8\lambda_1 - 2\lambda_2 + \lambda_3 - 6$}}(v24);
				\draw [postaction={decorate}] (v24)-- node[above, opacity=0.9] {\footnotesize{$-3\lambda_1 - 8\lambda_2 + 6$}}(v44);
				
				\draw [postaction={decorate}] (v00)-- node[opacity=0.9, rotate=90,yshift=0.25cm] {\footnotesize{$8\lambda_2 - 12$}}(v02);
				\draw [postaction={decorate}] (v02)-- node[opacity=0.9, rotate=90,yshift=0.25cm] {\footnotesize{$3\lambda_2 - 5\lambda_3 + 7$}}(v04);
				\draw [postaction={decorate}] (v20)-- node[opacity=0.9, rotate=90,yshift=0.25cm] {\footnotesize{$2\lambda_1 - 6\lambda_2 + \lambda_3 - 1$}}(v22);
				\draw [postaction={decorate}] (v22)-- node[opacity=0.9, rotate=90,yshift=0.25cm] {\footnotesize{$5\lambda_1 + 8$}}(v24);
				\draw [postaction={decorate}] (v40)-- node[opacity=0.9, rotate=90,yshift=-0.25cm] {\footnotesize{$3\lambda_1 - \lambda_3 + 6$}}(v42);
				\draw [postaction={decorate}] (v42)-- node[opacity=0.9, rotate=90,yshift=-0.25cm] {\footnotesize{$-9\lambda_2 + \lambda_3 + 4$}}(v44);
			\end{scope}

            \begin{scope} [xshift=3.25cm]
			{
                \node[vertex] (v00) at (0,0) [fill=blue,circle, draw, inner sep=0pt, minimum size=8pt, label=below left:{\Large $s$}] {};
                

                \node[vertex] (v02) at (0,2) [fill=blue,circle, draw, inner sep=0pt, minimum size=8pt] {};
                

                \node[vertex] (v04) at (0,4) [fill=blue,circle, draw, inner sep=0pt, minimum size=8pt] {};
                

                \node[vertex] (v20) at (2,0) [fill=blue,circle, draw, inner sep=0pt, minimum size=8pt] {};


                \node[vertex] (v22) at (2,2) [fill=blue,circle, draw, inner sep=0pt, minimum size=8pt] {};


                \node[vertex] (v24) at (2,4) [fill=blue,circle, draw, inner sep=0pt, minimum size=8pt] {};

				
				\node[vertex] (v40) at (4,0) [fill=blue,circle, draw, inner sep=0pt, minimum size=8pt] {};


                \node[vertex] (v42) at (4,2) [fill=blue,circle, draw, inner sep=0pt, minimum size=8pt] {};


                \node[vertex] (v44) at (4,4) [fill=blue,circle, draw, inner sep=0pt, minimum size=8pt, label=above right:{\Large $t$}] {};

			}
			\end{scope}
			
			\begin{scope} [decoration={markings, mark=at position 0.6 with {\arrow[scale=2,>=stealth,gray]{>}}}]
				\draw [postaction={decorate}] (v00)-- node[below, opacity=0.9] {\footnotesize{$-31$}} (v20);
				\draw [postaction={decorate}] (v20)-- node[below, opacity=0.9] {\footnotesize{$53$}}(v40);
				\draw [postaction={decorate}] (v22)-- node[above, opacity=0.9] {\footnotesize{$43$}}(v42);
				\draw [postaction={decorate}] (v04)-- node[above, opacity=0.9] {\footnotesize{$26$}}(v24);
				
				\draw [postaction={decorate}] (v02)-- node[opacity=0.9, rotate=90,yshift=0.25cm] {\footnotesize{$-26$}}(v04);
				\draw [postaction={decorate}] (v20)-- node[opacity=0.9, rotate=90,yshift=0.25cm] {\footnotesize{$17$}}(v22);
				\draw [postaction={decorate}] (v40)-- node[opacity=0.9, rotate=90,yshift=-0.25cm] {\footnotesize{$9$}}(v42);
				\draw [postaction={decorate}] (v42)-- node[opacity=0.9, rotate=90,yshift=-0.25cm] {\footnotesize{$19$}}(v44);
			\end{scope}

            \begin{scope} [decoration={markings, mark=at position 0.5 with {\arrow[scale=2,>=stealth,gray]{>}}}]
                    \draw [postaction={decorate}] (v22)-- node[opacity=0.9, rotate=90,yshift=0.25cm,xshift=0.15cm] {\footnotesize{$23$}}(v24);
                    \draw [postaction={decorate}] (v00)-- node[opacity=0.9, rotate=90,yshift=0.25cm,xshift=0.1cm] {\footnotesize{$-20$}}(v02);
                    \draw [postaction={decorate}] (v24)-- node[above, opacity=0.9,xshift=0.05cm] {\footnotesize{$5$}}(v44);
                    \draw [postaction={decorate}] (v02)-- node[above, opacity=0.9,xshift=0.15cm] {\footnotesize{$-25$}}(v22);
            \end{scope}

            \begin{scope} [decoration={markings, mark=at position 0.5 with {\arrow[scale=2,>=stealth,red]{>}}}, ultra thick, red]
                \draw [postaction={decorate}] (v22)-- node[opacity=0.9, rotate=90] {}(v24);

                \draw [postaction={decorate}] (v00)-- node[opacity=0.9, rotate=90] {}(v02);

                \draw [postaction={decorate}] (v24)-- node[above, opacity=0.9] {}(v44);

                \draw [postaction={decorate}] (v02)-- node[above, opacity=0.9] {}(v22);
            \end{scope}

		\end{tikzpicture}
		\caption{(Left) A directed graph $G$ with two special vertices $s$ and $t$, and edge weights of the form $a_{e,1} \lambda_1 + a_{e,2} \lambda_2 + a_{e,3} \lambda_3 + a_{e,4}$ (that is, $d=3$). The MSPC of this graph is explained by \Cref{fig:allpathsmcsp}. (Right) An instantiation of $G$ with $\lambda_1=3$, $\lambda_2=-1$, $\lambda_3=6$. The shortest path from $s$ to $t$ in this instantiation of $G$ is indicated in thick red.}\label{fig:planar3x3}
	\end{figure}
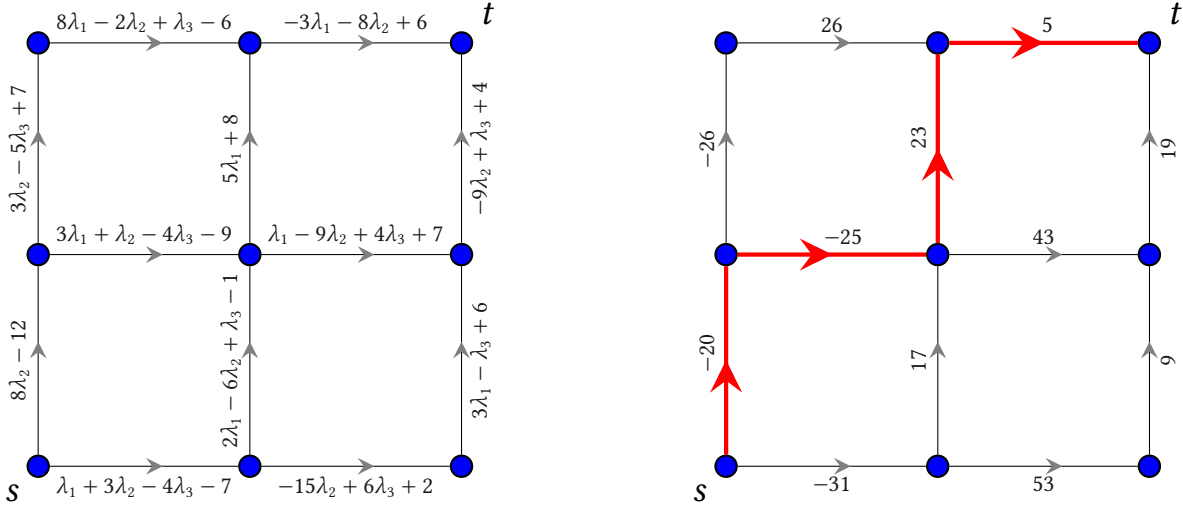


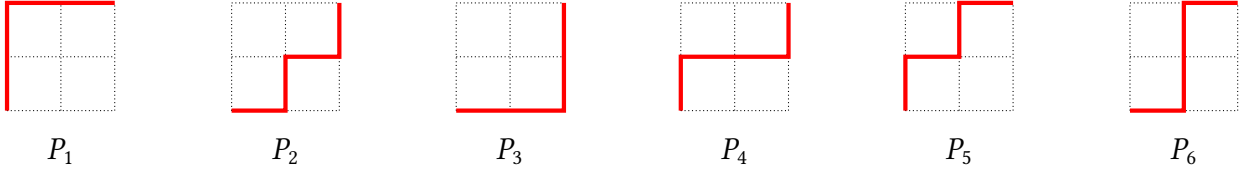
\begin{figure}[ht]
    \vspace{0.5cm}
	\centering
	\begin{tikzpicture}[scale=0.95]
        
        \def \k {0.75}
        \def \gap {4.17}
        
        \foreach \i in {0,...,5}
        {
            \draw[step=\k, xshift=\i*\gap*\k cm, densely dotted] (0,0) grid (2*\k,2*\k);
            \pgfmathparse{\i+1};
            \edef\j{\pgfmathresult};
            \node at (\i*\gap*\k+\k,-\k*\k) {\large{$P_{\pgfmathprintnumber{\j}}$}};
        }
        
        \draw[ultra thick, red, xshift=0*\gap*\k cm] (0,0)--(0,\k)--(0,2*\k)--(\k,2*\k)--(2*\k,2*\k);
        \draw[ultra thick, red, xshift=1*\gap*\k cm] (0,0)--(\k,0)--(\k,\k)--(2*\k,\k)--(2*\k,2*\k);
        \draw[ultra thick, red, xshift=2*\gap*\k cm] (0,0)--(\k,0)--(2*\k,0)--(2*\k,\k)--(2*\k,2*\k);
        \draw[ultra thick, red, xshift=3*\gap*\k cm] (0,0)--(0,\k)--(\k,\k)--(2*\k,\k)--(2*\k,2*\k);
        \draw[ultra thick, red, xshift=4*\gap*\k cm] (0,0)--(0,\k)--(\k,\k)--(\k,2*\k)--(2*\k,2*\k);
        \draw[ultra thick, red, xshift=5*\gap*\k cm] (0,0)--(\k,0)--(\k,\k)--(\k,2*\k)--(2*\k,2*\k);
        
	\end{tikzpicture}
    \vspace{-0.2cm}
	\caption{The set $\mathcal{P}_{s,t} = \{P_1,P_2,P_3,P_4,P_5,P_6\}$ of all paths from $s$ to $t$ in the graph $G$ in \Cref{fig:planar3x3} (left). Note that $P_5$ is the only shortest path in $G$ at $\overline{\lambda}=(3,-1,6)$, as shown in \Cref{fig:planar3x3} (right). So $P_5\in \PSP(G)$. Similarly, $P_2$, $P_4$, $P_6$ are the only shortest paths in $G$ at $(0,0,-100)$, $(0,0,0)$, $(-200,2,0)$, respectively. So $P_2,P_4,P_6\in\PSP(G)$. Also, $P_1,P_3\notin\PSP(G)$, because $\mathsf{wt}(P_1)=\mathsf{wt}(P_5)+2$ and $\mathsf{wt}(P_3)=\mathsf{wt}(P_2)+2$. Hence, $|\PSP(G)|=4$.}\label{fig:allpathsmcsp}
    \end{figure}

\subsection{Prior Work}

Carstensen showed\footnote{Carstensen attributed the proof of this upper bound to Gusfield.} that every graph $G$ on $n$ vertices whose edge weights are of the form $\mathsf{wt}(e)=a_{e,1}\lambda_1+a_{e,2}$ (that is, $d=1$) admits an MSPC of size $n^{O(\log n)}$. She also exhibited graphs whose MSPC is of size $n^{\Omega(\log n)}$, proving that her upper bound is optimal up to the constant in the exponent.

Building upon Carstensen's work, Mulmuley \& Shah~\cite{Mulmuley_Shah_2000} proved that her $n^{\Omega(\log n)}$ lower bound can be expressed with edge weights having just poly-logarithmic bits (that is, the $a_{e,i}$'s for each edge $e$ can be represented using $\log^{O(1)}(n)$ bits). Gajjar \& Radhakrishan~\cite{Gajjar_Radhakrishnan_2019} modified their construction and edge weights slightly (still keeping them poly-logarithmic), gave a better exposition of their proof, and showed that the $n^{\Omega(\log n)}$ lower bound also holds (possibly with a different constant in the exponent) for planar graphs, thereby refuting a conjecture of Nikolova~\cite{Nikolova2009}.

Gajjar \& Radhakrishan~\cite{Gajjar_Radhakrishnan_2019}  also explored graphs with edge weights of the form $a_{e,1}\lambda_1 + a_{e,2}\lambda_2 + a_{e,3}$ (that is, $d=2$), and proved an upper bound of $n^{O(\log^2 n)}$ on the size of the MSPC in such graphs. However, their proof did not work for three parameters and beyond. Barth, Funke \& Proissl~\cite{Barth_Funke_Proissl_2022} succeeded in extending their idea to $d$ parameters (graphs with edge weights of the form $a_{e,1} \lambda_1 + a_{e,2} \lambda_2 + a_{e,3} \lambda_3 + \cdots + a_{e,d} \lambda_d + a_{e,d+1}$), showing an upper bound of $n^{O_d(\log^d n)}$ for all positive integers $d$. Chatterjee, Gajjar \& Radhakrishnan~\cite{CGRICTS} improved this\footnote{
The constant behind the big-O notation in the exponent of~\cite{Barth_Funke_Proissl_2022} is $2^d$ (an exponential dependence on $d$). This was improved to an absolute constant by~\cite{CGRICTS} (their constant is simply the number $4$, and therefore has no dependence on $d$).} to $n^{O(\log^d n)}$, which is state-of-the-art.

The lower bound, however, has remained $n^{\Omega(d \log n)}$. Though not explicitly mentioned in any of the preceding papers, this lower bound can be easily realized by considering a graph with edge weights of the form $a_{e,1}\lambda_1 + a_{e,2}$ whose MSPC is of size $n^{\Omega(\log n)}$, and then attaching $d$ copies of this graph with itself in series (for more details, see \Cref{clm:linearlowerbound}).

Note that the $d$ is in the (single) exponent in the lower bound, and in the double exponent (exponent of the exponent) in the upper bound. This leaves a massive gap between the lower and upper bounds. In this work, we fully bridge this gap.

\subsection{Our Contributions}

Our main contribution is an upper bound on the MSPC of directed acyclic graphs with linear edge weights.
\begin{restatable}{theorem}{linearupperbound}\label{thm:linear_upper_bound}
    Let $G = (V,E)$ be an $n$-vertex directed acyclic graph with edge weights of the form $a_{e,1} \lambda_1 + a_{e,2} \lambda_2 + \cdots + a_{e,d} \lambda_d + a_{e,d+1}$. Then,
    \begin{equation*}
        |\PSP(G)| \in n^{O(d\log{n})}.
    \end{equation*}
\end{restatable}

We also show that our proof for \Cref{thm:linear_upper_bound} can be adapted to work for directed graphs without negative-weight cycles, and for undirected graphs without negative-weight edges.

\begin{restatable}{theorem}{directedgraphsnonegcyc}\label{thm:directedgraphsnonegcyc}
    Let $G = (V,E)$ be an $n$-vertex directed  graph without negative-weight cycles with edge weights of the form $a_{e,1} \lambda_1 + a_{e,2} \lambda_2 + \cdots + a_{e,d} \lambda_d + a_{e,d+1}$. Then,
    \begin{equation*}
        |\PSP(G)| \in n^{O(d\log{n})}.
    \end{equation*}
\end{restatable}

\begin{restatable}{theorem}{undirectedgraphsnonegcedg}\label{thm:undirectedgraphsnonegedg}
    Let $G = (V,E)$ be an $n$-vertex undirected graph without negative-weight edges with edge weights of the form $a_{e,1} \lambda_1 + a_{e,2} \lambda_2 + \cdots + a_{e,d} \lambda_d + a_{e,d+1}$. Then,
    \begin{equation*}
        |\PSP(G)| \in n^{O(d\log{n})}.
    \end{equation*}
\end{restatable}

We also show an upper bound on the MSPC of directed graphs with univariate polynomial edge weights.

\begin{restatable}{theorem}{univardegd}\label{thm:univardegd}
    Let $G = (V,E)$ be an $n$-vertex directed graph with edge weights of the form $a_{e,q} \lambda^q + a_{e,q-1} \lambda^{q-1} + \cdots + a_{e,2} \lambda^2 + a_{e,1} \lambda + a_{e,0}$. Then,
    \begin{equation*}
        |\PSP(G)| \in n^{O(\log{n}+\log{q})}.
    \end{equation*}
\end{restatable}

We also study the problem from the algorithmic standpoint. To this end, we construct a data structure that preprocesses the graph, and efficiently computes a shortest path in the graph for a given $\overline{\lambda}$ in real time.

\begin{restatable}{theorem}{sporaclelin}\label{thm:sporaclelin}
    Let $G = (V,E)$ be an $n$-vertex directed graph with edge weights of the form $a_{e,1} \lambda_1 + a_{e,2} \lambda_2 + \cdots + a_{e,d} \lambda_d + a_{e,d+1}$. Then, there exists a data structure that takes as input an $\bar{x}\in\mathbb{R}^d$, and outputs a shortest path in the graph $G$ at $\overline{\lambda}=\overline{x}$ in $O(d^4\log^2 n)$ time. The data structure takes $n^{\tilde{O}(d)}$ space and $n^{\tilde{O}(d)}$ preprocessing time.
\end{restatable}

\begin{restatable}{theorem}{sporacledeg}\label{thm:sporacledeg}
    Let $G = (V,E)$ be an $n$-vertex directed graph with edge weights of the form $a_{e,q} \lambda^q + a_{e,q-1} \lambda^{q-1} + \cdots + a_{e,2} \lambda^2 + a_{e,1} \lambda + a_{e,0}$. Then, there exists a data structure that takes as input an $x\in\mathbb{R}$, and outputs a shortest path in the graph $G$ at $\lambda=x$ in $O(\log^2 n + \log n \log q)$ time. The data structure takes $(nq)^{\tilde{O}(1)}$ space and $(nq)^{\tilde{O}(1)}$ preprocessing time.
\end{restatable}


\subsection{Proof Overview}

In this section, we outline in detail the main ideas behind the proof of our main result (\Cref{thm:linear_upper_bound}). The proofs of our other results can be easily followed once this is understood.

Let $G = G_{n,n}$ be a layered DAG with $n$ layers, $n$ vertices per layer, source $s$ and sink $t$. Let the edge weights be linear functions of $d$ parameters $\lambda_1,\ldots, \lambda_d$. To bound $|\PSP(G)|$, a  standard divide-and-conquer approach does the following -- split the graph into $2n$ instances of $G_{n/2, n}$ through the vertices of the middle layer $\{v_1, \ldots, v_n\}$ and then recurse on those. In particular, each of these $G_{n/2,n}$ instances are of the form source $s$ and sink $v_i$ (call this $G_i$) for each $i$, or of the form source $v_i$ and sink $t$ (call this $G'_i$) for each $i$. 

If graphs $G$ and $G'$ are connected in series then $|\PSP(G\circ_{\text{ser}} G')|\leq |\PSP(G)|\cdot|\PSP(G')|$, and if they are connected in parallel, $|\PSP(G\circ_{\text{par}} G')|\leq |\PSP(G)|+|\PSP(G')|$.

\begin{align*}
    |\PSP(G)| \leq \sum_{i=1}^n |\PSP(G_{i})|\cdot |\PSP(G'_{i})|.
\end{align*}

Instead of dividing this at the middle layer, one could generalize the aforementioned divide-and-conquer argument through  $\ell$ such layers and then merge carefully. We note that even this carefully divided approach does not give bounds that are better than $n^{O\left(\log^d(n)/\log\log^{d-1}(n)\right)}$ (for $d\geq 2$). Though this is a better bound than that of \cite{Barth_Funke_Proissl_2022}, it is only a modest improvement. 

We first observe that when two graphs $G_1$ and $G_2$ are connected in series, the size of the shortest path cover may be much less than the product of the sizes of the shortest path covers in these individual instances. That is, a parametric shortest path $P$ in $G_1$ concatenates with a parametric shortest path $P'$ in $G_2$ if and only if there is a point in $\mathbb{R}^d$ where these are simultaneously shortest in their respective graphs. We illustrate this through the following example. 

Let $G$ be a directed acyclic graph as shown in \cref{fig:simpleAlmondGraph} whose edge weights vary as linear functions of parameters $\lambda_1$ and $\lambda_2$. Note that edge 0 between $s$ and $v$ would be the shortest edge for all values of $(\lambda_1, \lambda_2)\in\mathbb{R}^2$ such that $\lambda_1+2\lambda_2\leq 2\lambda_1 + \lambda_2$,  and edge 1 would be the shortest edge for all values of $(\lambda_1, \lambda_2)\in\mathbb{R}^2$ such that $\lambda_1+2\lambda_2 \geq 2\lambda_1 + \lambda_2$ (illustrated in \cref{fig:pspAG_s2v}). Let $\pi_{s,v}$ denote this partition of $\mathbb{R}^2$ through the line $\lambda_1 = \lambda_2$.  Similary, edge 0 between $v$ and $t$ would be the shortest edge for all values of $(\lambda_1, \lambda_2)\in\mathbb{R}^2$ such that $\lambda_1+\lambda_2\leq  2\lambda_1 + \lambda_2 - 3$,  and edge 1 would be the shortest edge for all values of $(\lambda_1, \lambda_2)\in\mathbb{R}^2$ such that $\lambda_1+ \lambda_2 \geq 2\lambda_1 + \lambda_2 - 3$ (illustrated in \cref{fig:pspAG_v2t}). Let $\pi_{v,t}$ denote this partition of $\mathbb{R}^2$ through the line $\lambda_1 = 3$. 

We represent each $s$ to $t$ path using a tuple $(a,b)$ (for $a,b\in\{0,1\}$) to indicate that the path takes edge $a$ between $s$ and $v$, and edge $b$ between $v$ and $t$. With this representation, there are four such paths $(0,0), (0,1), (1,0), (1,1)$ with parametric weights $2\lambda_1 + 3\lambda_2$, $3\lambda_1 + 3\lambda_2 - 3$, $3\lambda_1 + 2\lambda_2$, $4\lambda_1 + 2\lambda_2 - 3$ respectively. In particular, path $(0,0)$ is shortest for all values of $(\lambda_1, \lambda_2)$ that simultaneously satisfy the inequalities
\begin{align*}
    2\lambda_1 + 3\lambda_2 &\leq 3\lambda_1 + 3\lambda_2 - 3,\\
    2\lambda_1 + 3\lambda_2 &\leq 3\lambda_1 + 2\lambda_2,\\
    2\lambda_1 + 3\lambda_2 &\leq 4\lambda_1 + 2\lambda_2.
\end{align*}
which simplify to $\lambda_1 \geq 3$ and $\lambda_1\geq \lambda_2$. These set of inequalities defines a region in $\mathbb{R}^2$ where path $(0,0)$ is the shortest. Similarly, we can define regions for the other paths. This creates a partition of the space $\pi_{s,t}$. We will now claim that such a partition could simply be obtained through a \emph{superimposition} of the partitions $\pi_{s,v}$ and $\pi_{v,t}$ by capitalizing on the independence of paths between the nodes $s$ and $v$, and nodes $v$ and $t$ (as illustrated in \cref{fig:pspAG_s2t}). Here, we also use that fact that in a DAG if a shortest $s$ to $t$ path passes through $v$, its segment between $s$ and $v$, and its segment between $v$ and $t$ both have to be shortest by themselves.

More formally, the lines $\lambda_1 = \lambda_2$ and $\lambda_1 = 3$ together partition $\mathbb{R}^2$ into four regions and each of these regions corresponds to one of the $s$ to $t$ paths. For a given point, its side with respect $\lambda_1 = \lambda_2$ helps us identify which of the two edges between $s$ and $v$ is the shortest, and its side with respect to $\lambda_1 = 3$ helps us identify which of the two edges between $v$ and $t$ is the shortest. 

\begin{figure}
\centering
\begin{subfigure}[b]{\textwidth}
\centering
    \begin{tikzpicture}[
    node distance=3cm,
    mynode/.style={circle, draw, thick, minimum size=8pt, inner sep=2pt}
]
    \node[mynode] (S) at (0,0) {$s$};
    \node[mynode] (V) at (4,0) {$v$};
    \node[mynode] (T) at (8,0) {$t$};
    \tikzset{
        myarrow/.style={->, very thick}
    }
    \draw[myarrow] (S) to [bend left=50] 
        node[pos=0.5, above, yshift=6pt, font=\footnotesize] {$\lambda_1 + 2\lambda_2$} 
        node[pos=0.5, below, yshift=-3pt, font=\footnotesize, elegantgray] {\textsf{Edge 0}} (V);
    \draw[myarrow] (S) to [bend right=50] 
        node[pos=0.5, below, yshift=-6pt, font=\footnotesize] {$2\lambda_1 + \lambda_2$} 
        node[pos=0.5, above, yshift=3pt, font=\footnotesize, elegantgray] {\textsf{Edge 1}} (V);
    \draw[myarrow] (V) to [bend left=50] 
        node[pos=0.5, above, yshift=6pt, font=\footnotesize] {$\lambda_1 + \lambda_2$} 
        node[pos=0.5, below, yshift=-3pt, font=\footnotesize, elegantgray] {\textsf{Edge 0}} (T);
    \draw[myarrow] (V) to [bend right=50] 
        node[pos=0.5, below, yshift=-6pt, font=\footnotesize] {$2\lambda_1 + \lambda_2 - 3$} 
        node[pos=0.5, above, yshift=3pt, font=\footnotesize, elegantgray] {\textsf{Edge 1}} (T);
\end{tikzpicture}
\caption{Graph $G$ with linear edge weights. Note the the edges above are denoted \textsf{Edge 0}, \textsf{Edge 1}, \textsf{Edge 0}, \textsf{Edge 1}. Correspondingly, the four paths from $s$ to $t$ are denoted \textsf{Path (0,0)}, \textsf{Path (0,1)}, \textsf{Path (1,0)}, \textsf{Path (0,1)}.}
\label{fig:simpleAlmondGraph}
\end{subfigure}

\begin{subfigure}[b]{0.35\textwidth}
    \vspace{0.5cm}
    \centering
    \begin{tikzpicture}[scale=0.25]
    \clip (-12.5,-12.5) rectangle (15,15);
    
    \begin{scope}
    \clip (-10,-10) rectangle (10,10);
    \fill[elegantblue!30] (-10,-10) -- (10,10) -- (10,-10) -- cycle;
    \fill[elegantcoral!30] (-10,-10) -- (10,10) -- (-10,10) -- cycle;
    \end{scope}
    
    \foreach \x in {-10,-5,...,10}
        \draw[elegantgray!15, very thin] (\x,-10) -- (\x,10);
    \foreach \y in {-10,-5,...,10}
        \draw[elegantgray!15, very thin] (-10,\y) -- (10,\y);
    
    \draw[->, elegantgray!50, thick] (-10,0) -- (11,0) node[right, black, font=\large] {$\lambda_1$};
    \draw[->, elegantgray!50, thick] (0,-10) -- (0,11) node[above, black, font=\large] {$\lambda_2$};
    
    
    \draw[elegantlavender!80, line width=2.5pt] (-10,-10) -- (10,10);
    \node at (7, 8) [elegantlavender!80!black, font=\normalsize, rotate=45] {$\lambda_2 = \lambda_1$};
    
    
    \node at (-5, 5.5) [font=\large, elegantcoral!80!black] {\textsf{Edge 1}};
    \node at (-5, 3.5) [font=\tiny, elegantgray] {$2\lambda_1 + \lambda_2<\lambda_1 + 2\lambda_2$};
    
    \node at (5, -5.5) [font=\large, elegantblue!80!black] {\textsf{Edge 0}};
    \node at (5, -3.5) [font=\tiny, elegantgray] {$2\lambda_1 + \lambda_2>\lambda_1 + 2\lambda_2$};
\end{tikzpicture}
    \caption{Partition of $\mathbb{R}^2$ to indicate parametric shortest paths between $s$ and $v$}
    \label{fig:pspAG_s2v}
\end{subfigure}\qquad\qquad\qquad
\begin{subfigure}[b]{0.35\textwidth}
    \begin{tikzpicture}[scale=0.25]
    \clip (-12.5,-12.5) rectangle (15,15);
    
    \begin{scope}
    \clip (-15,-15) rectangle (15,15);
    \fill[elegantsage!30] (3,-10) -- (10,-10) -- (10,10) -- (3,10) -- cycle;
    \fill[elegantgold!30] (-10,-10) -- (3,-10) -- (3,10) -- (-10,10) -- cycle;
    \end{scope}
    
    \foreach \x in {-10,-5,...,10}
        \draw[elegantgray!15, very thin] (\x,-10) -- (\x,10);
    \foreach \y in {-10,-5,...,10}
        \draw[elegantgray!15, very thin] (-10,\y) -- (10,\y);
    
    \draw[->, elegantgray!50, thick] (-10,0) -- (11,0) node[right, black, font=\large] {$\lambda_1$};
    \draw[->, elegantgray!50, thick] (0,-10) -- (0,11) node[above, black, font=\large] {$\lambda_2$};
    
    
    \draw[elegantlavender!80, line width=2.5pt] (3,-10) -- (3,10);
    \node at (5, 8) [elegantlavender!80!black, font=\small, rotate=90] {$\lambda_1 = 3$};
    
    
    \node at (-5, 5.5) [font=\large, elegantgold!80!black] {\textsf{Edge 1}};
    \node at (-5, 3.5) [font=\tiny, elegantgray] {$2\lambda_1 + \lambda_2 - 3 < \lambda_1 + \lambda_2$};
    
    \node at (8, -5.5) [font=\large, elegantsage!80!black, rotate = 90] {\textsf{Edge 0}};
    \node at (6, -4.5) [font=\tiny, elegantgray, rotate = 90] {$2\lambda_1 + \lambda_2 - 3 > \lambda_1 + \lambda_2$};
\end{tikzpicture}
\caption{Partition of $\mathbb{R}^2$ to indicate parametric shortest paths between $v$ and $t$}
    \label{fig:pspAG_v2t}
\end{subfigure}
\begin{subfigure}[b]{\textwidth}
\vspace{0.75cm}
\centering
\begin{tikzpicture}[scale=0.25]
    \clip (-12.5,-12.5) rectangle (15,15);
    
    \begin{scope}
    \clip (-10,-10) rectangle (10,10);
    \fill[elegantblue!35] (-10,-10) -- (3,-10) -- (3,3) -- (-10,-10) -- cycle;
    \fill[elegantcoral!35] (-10,-10) -- (-10,10) -- (3,10) -- (3,3) -- cycle;
    \fill[elegantgold!35] (3,3) -- (3,-10) -- (10,-10) -- (10,10) -- cycle;
    \fill[elegantsage!35] (3,3) -- (3,10) -- (10,10) -- cycle;
    \end{scope}
    
    \foreach \x in {-10,-5,...,10}
        \draw[elegantgray!10, very thin] (\x,-10) -- (\x,10);
    \foreach \y in {-10,-5,...,10}
        \draw[elegantgray!10, very thin] (-10,\y) -- (10,\y);
    
    \draw[->, elegantgray!50, thick] (-10,0) -- (11,0) node[right, black, font=\large] {$\lambda_1$};
    \draw[->, elegantgray!50, thick] (0,-10) -- (0,11) node[above, black, font=\large] {$\lambda_2$};
    
    
    \draw[elegantlavender!70, line width=2.5pt] (-10,-10) -- (10,10);
    \node at (9.7, 8.5) [elegantlavender!80!black, font=\normalsize, rotate=45] {$\lambda_2 = \lambda_1$};
    
    \draw[elegantlavender!70, line width=2.5pt] (3,-10) -- (3,10);
    \node at (2.2, -8.5) [elegantlavender!80!black, font=\normalsize, rotate=90] {$\lambda_1 = 3$};
    
    
    \node at (-2.6, -6.5) [font=\large, elegantblue!80!black] {\textsf{Path (0,1)}};
    \node at (-2.6, -8) [font=\footnotesize, elegantgray] {$3\lambda_1 + 3\lambda_2 - 3$};
    
    \node at (-4.5, 5) [font=\large, elegantcoral!80!black] {\textsf{Path (1,1)}};
    \node at (-4.5, 3.5) [font=\footnotesize, elegantgray] {$4\lambda_1 + 2\lambda_2 - 3$};
    
    \node at (6.5, -3.5) [font=\large, elegantgold!80!black] {\textsf{Path (0,0)}};
    \node at (6.5, -5) [font=\footnotesize, elegantgray] {$2\lambda_1 + 3\lambda_2$};
    
    \node at (6.5, 12.5) [font=\large, elegantsage!80!black] {\textsf{Path (1,0)}};
    \node at (6.5, 10.7) [font=\footnotesize, elegantgray] {$3\lambda_1 + 2\lambda_2$};
    
    \fill[elegantlavender!80] (3,3) circle (6pt);
    \draw[white, thick] (3,3) circle (6pt);
\end{tikzpicture}
\caption{Overlaying (superimposing) the partitions in (b) and (c)}
    \label{fig:pspAG_s2t}
\end{subfigure}
\caption{Partition of $\mathbb{R}^2$ by hyperplanes created by regions corresponding to shortest paths from $s$ to $t$ in $G$}
\end{figure}
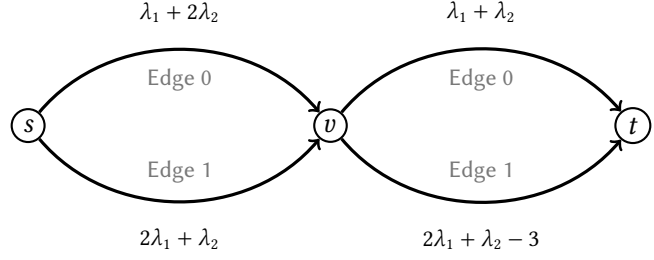
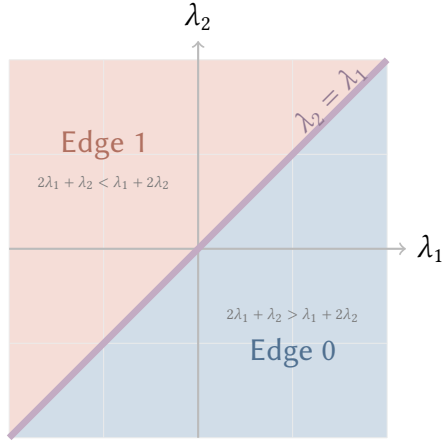
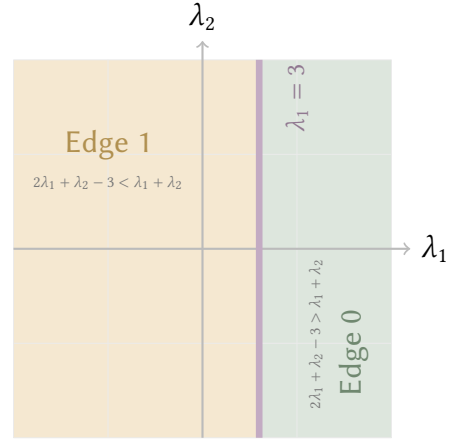
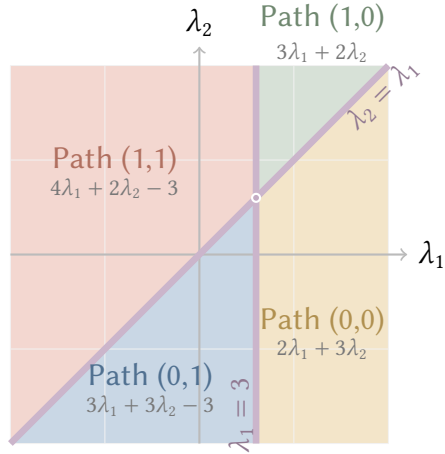

\begin{figure}[!h]
    \centering
    \begin{subfigure}[b]{0.35\textwidth}
\begin{tikzpicture}[
    node distance=0.75cm,
    vertex/.style={circle, draw, thick, minimum size=20pt, inner sep=0pt, fill=white} 
]

    \node[vertex] (mid0) {$s$};
    \node[vertex] (mid2) [right=3.8cm of mid0] {$t$}; 

    \node[vertex] (v_i) [right=1.5cm of mid0, yshift=0cm] {$u_i$};
    
    \node (dots_top) [above=0.2cm of v_i, inner sep=0pt] {$\vdots$};
    \node (dots_bottom) [below=0.2cm of v_i, inner sep=0pt] {$\vdots$};

    \node[vertex] (v2) [above=0.2cm of dots_top] {$u_2$};
    \node[vertex] (v1) [above=0.2cm of v2] {$u_1$};
    \node[vertex] (vnminus1) [below=0.2cm of dots_bottom] {$u_{n-1}$};
    \node[vertex] (vn) [below=0.2cm of vnminus1] {$u_n$};

    \path[->, thick]
    (mid0) edge (v1)
    (mid0) edge (v2)
    (mid0) edge (v_i)
    (mid0) edge (vnminus1)
    (mid0) edge (vn)
    
    (v1) edge (mid2)
    (v2) edge (mid2)
    (v_i) edge (mid2)
    (vnminus1) edge (mid2)
    (vn) edge (mid2);

\end{tikzpicture}

    \caption{Graph $G_1$ of length 2}
    \label{fig:overview_gadget1}
    \end{subfigure}\quad\quad
    \begin{subfigure}[b]{0.35\textwidth}
\begin{tikzpicture}[
    node distance=0.75cm,
    vertex/.style={circle, draw, thick, minimum size=20pt, inner sep=0pt, fill=white} 
]

    \node[vertex] (mid0) {$s'$};
    \node[vertex] (mid2) [right=3.8cm of mid0] {$t'$}; 

    \node[vertex] (v_i) [right=1.5cm of mid0, yshift=0cm] {$v_i$};
    
    \node (dots_top) [above=0.2cm of v_i, inner sep=0pt] {$\vdots$};
    \node (dots_bottom) [below=0.2cm of v_i, inner sep=0pt] {$\vdots$};

    \node[vertex] (v2) [above=0.2cm of dots_top] {$v_2$};
    \node[vertex] (v1) [above=0.2cm of v2] {$v_1$};
    \node[vertex] (vnminus1) [below=0.2cm of dots_bottom] {$v_{n-1}$};
    \node[vertex] (vn) [below=0.2cm of vnminus1] {$v_n$};

    \path[->, thick]
    (mid0) edge (v1)
    (mid0) edge (v2)
    (mid0) edge (v_i)
    (mid0) edge (vnminus1)
    (mid0) edge (vn)
    
    (v1) edge (mid2)
    (v2) edge (mid2)
    (v_i) edge (mid2)
    (vnminus1) edge (mid2)
    (vn) edge (mid2);

\end{tikzpicture}
    \caption{Graph $G_2$ of length 2}
    \label{fig:overview_gadget2}
    \end{subfigure}
    \caption{Disjoint graphs $G_1$ and $G_2$}
    \label{fig:overview_gadgets}
\end{figure}
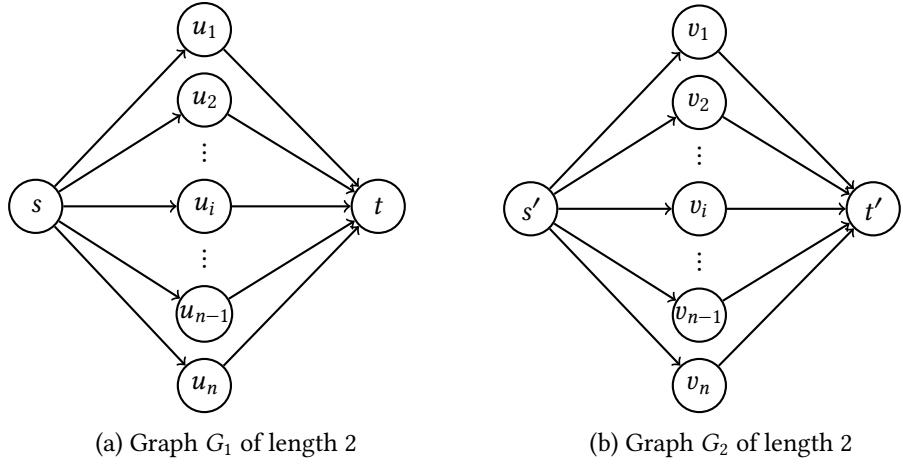

Let $G_1$ and $G_2$ be two disjoint graphs as shown in \cref{fig:overview_gadgets}. In graph $G_1$, nodes $u_1, \ldots, u_n$ are $n$ distinct degree-$2$ nodes that connect $s$ and $t$ by providing $n$ disjoint paths. For each $i\in[n]$, let $L_i(\lambda_1,\ldots, \lambda_d)$ be the linear weight function associated with the path $s\rightarrow u_i \rightarrow t$. Similarly, in graph $G_2$, nodes $v_1, \ldots, v_n$ are $n$ distinct degree-$2$ nodes that connect $s'$ and $t'$ by providing $n$ disjoint paths. For each $i\in[n]$, let $L'_i(\lambda_1,\ldots, \lambda_d)$ be the linear weight function associated with the path $s'\rightarrow v_i \rightarrow t'$. 

Through the generalization of the afore mentioned discussion, for each $i\in [n]$ we can infer that the set of values of $(\lambda_1, \ldots, \lambda_d)$ from $ \mathbb{R}^d$ for which the path $s\rightarrow u_i \rightarrow t$ in graph $G_1$ is the shortest is specified by the set of linear inequalities $\mathcal{J}_i = \{L_i - L_j \leq 0 \mid j\in [n]\setminus {i} \}$\footnote{If the set of inequalities are inconsistent then the region would be empty. For example, with $L_1 = \lambda_1 - \lambda_2, L_2 = \lambda_2 - \lambda_1, L_3 = 1$, path 3 can never be a shortest path, as both $\lambda_1 - \lambda_2 \geq 1$ and $\lambda_2 - \lambda_1 \geq 1$ cannot be satisfied simultaneously.}. Let the corresponding hyperplanes that help create this region be $\mathcal{H}_i = \{ L_i - L_j = 0 \mid j\in [n]\setminus {i} \}$. For the sake of exposition, let us assume that all our sets of inequalities are consistent. 

We first note that these hyperplanes $\mathcal{H} = \mathcal{H}_1\cup \cdots\cup \mathcal{H}_n$ partition\footnote{That is, every point in $\mathbb{R}^d$ belongs to at least one of the regions $\mathcal{J}_1, \ldots, \mathcal{J}_n$ (assuming they are all consistent). For any arbitrary $\bar{a}\in \mathbb{R}^d$, compute the values $L_1(\bar{a}), \ldots, L_n(\bar{a})$ and take the $\mathrm{argmin}$ of these values. Call it $i^*$. It is easy to see that $\bar{a}$ simultaneously satisfies the set of inequalities $\mathcal{J}_{i^*}$ and lies in the region defined by it.} $\mathbb{R}^d$ such that each region thus created corresponds to exactly one path from $s$ to $t$. Let us denote this partition by $\pi_{s,t}$. Note that a parametric shortest path can appear with multiplicity in more than one contiguous regions created by these hyperplanes. That is, we are creating a map from the regions created through the partition of $\mathbb{R}^d$ with the set of hyperplanes $\mathcal{H}$, to the parametric shortest paths. Analogously, we can obtain a partition $\pi_{s',t'}$ with the set of hyperplanes $\mathcal{H}'$ such that each region in it corresponds to a parametric shortest path between $s'$ and $t'$. Note that the cardinality of  $\mathcal{H}\cup\mathcal{H}'$ is at most $2\binom{n}{2}$.

Similar to our earlier analysis, by superimposing the partitions $\pi_{s,t}$ and $\pi_{s',t'}$, we get a new partition in which the regions simultaneously indicate the shortest paths between $s\rightsquigarrow t$ and $s' \rightsquigarrow t'$. It is important to note that the above analysis also holds for any graphs $G_1$ and $G_2$ (even when they share edges or vertices). 

In case $G_1$ and $G_2$ were connected in series (with $t = s'$, source $s$ and sink $t'$), the total number of $s\rightsquigarrow t'$ paths are $n^2$ many and we could have easily formed hyperplanes by comparing each path with the rest simultaneously. This could potentially create $\binom{n^2}{2}$ many hyperplanes and this is far more than the count of the regions obtained through superimposition. 

Having established the superimposition principle for pairs of graphs, we now apply this technique recursively to a specific family of layered DAGs. This will allow us to derive concrete bounds on the number of parametric shortest paths. 

Assume that each layer is a complete bipartite graph and the edge weights vary as functions of $d$ parameters $\lambda_1, \ldots, \lambda_d$. Let $B$ be a set of pairs of vertices $\{(x,y)\mid x\in \text{Layer}_{2i-1}, y\in\text{Layer}_{2i+1}~\text{for}~i\in\{1, \ldots, \lfloor\frac{n-1}{2}\rfloor\} \}$. That is, $B$ contains ordered pairs of vertices across alternate odd layers. For each pair $(x,y)\in B$, there is a subgraph $G_{x,y}$ of length $2$ with $x$ as source and $y$ as sink, and graph $G_{x,y}$ resembles the graphs $G_1, G_2$ (\cref{fig:overview_gadgets}) in structure. Using the afore mentioned analysis, for each pair $(x,y) \in B$ we get a partition $\pi_{x,y}$ of $\mathbb{R}^d$ (through a set of hyperplanes $\mathcal{H} = \mathcal{H}_{x,y}$) such that every region in it corresponds to a parametric shortest path in $G_{x,y}$. By superimposing the partitions $\pi_{x,y}$ for all $(x,y) \in B$, we get finer regions of $\mathbb{R}^d$ such that each region \emph{uniquely} identifies a parametric shortest path between all the pairs $(x,y)$. In other words, these finer regions of $\mathbb{R}^d$ are generated by partitioning it with the set of hyperplanes $$\mathcal{H} = \bigcup_{(x,y)\in B} \mathcal{H}_{x,y}\,.$$ For each region thus created, we create a new DAG instance $G^{(x,y)}_{n,n/2}$ from $G_{n,n}$, with the following properties. Delete all vertices in even layers, along with the incident edges. Connect each pair $(x,y) \in B$ using a new edge such that the weight of this new edge $(x,y)$ is the cumulative weight of the shortest $x\rightsquigarrow y$ path identified through this region. 

For each region created through the hyperplanes, we create a new instance of the parametric shortest path problem on a graph of length $n/2$. The number of new instances created is equal to the number of regions generated by the hyperplanes $\mathcal{H}$ and this quantity can be bound efficiently (see \cref{thm:space_partitioning_theorem}). We recursively partition $\mathbb{R}^d$ through these new instances until we reach a trivial base case. Thus, 
\begin{align*}
    |\PSP(G_{n,n})| \leq \left(\text{\# of regions created by}~\mathcal{H}\right)\cdot \max_{(x,y)\in B}\left(|\PSP(G^{(x,y)}_{n,n/2})|\right)\,.
\end{align*}

In this paper, we consider edge weights that either vary as linear functions of $d$ real-valued parameters, and also edge weights that are univariate polynomials of degree at most $d$. The proof technique for the latter is not very different from the one for the former. 

We first present our proof for directed acyclic graphs (DAGs) (\Cref{thm:linear_upper_bound}). 
The result for DAGs can be lifted to count the size of an MSPC over a \emph{feasible}\footnote{A feasible region is the subset of $\mathbb{R}^d$ (or the set of values of $\overline{\lambda} = (\lambda_1,\lambda_2\ldots,\lambda_d)$) wherein the given directed graph (respectively, undirected graph) contains no negative-weight cycles (respectively, negative-weight edges). See \Cref{def:feasibleregion} for a formal definition.} region $S$, via a reduction to:

\begin{itemize}
    \item directed graphs with no negative-weight cycles 
    (\Cref{thm:directedgraphsnonegcyc}), and
    \item undirected graphs with no negative-weight edges 
    (\Cref{thm:undirectedgraphsnonegedg}).
\end{itemize}
More details of these reductions can be found in \cref{subsec:reduction2DAG}.

\section{Preliminaries}\label{sec:preliminaries}

In this section, we establish the terminologies and definitions that will be used throughout this work. 

\subsection{Basic Notation}

For a positive integer $n \in \mathbb{N}$, we use the notation $[n]$ to denote the set $\{1,2,\ldots,n\}$. Unless otherwise explicitly stated, all graphs considered in this work are simple (containing no multi-edges or self-loops), directed, and connected, with vertex set $V$ where $|V| = n$. All logarithms considered in this paper are taken to the base $2$ unless explicitly stated otherwise.

When working with parameter vectors, we employ the notation $\overline{\lambda} = (\lambda_1,\lambda_2\ldots,\lambda_d)$ to represent a $d$-tuple of real-valued parameters $\lambda_i \in \mathbb{R}$ for $i \in [d]$. For a subset $S \subseteq \mathbb{R}^d$ of the parameter space, we define $\PSP_{S}(G)$ as the collection of parametric shortest paths from a designated source vertex $s$ to a sink vertex $t$ when the parameter vector takes values in $S$. As a notational convenience, when $S = \mathbb{R}^d$ (i.e., when considering the entire parameter space), we use the abbreviated notation $\PSP(G) := \PSP_{\mathbb{R}^d}(G)$. In cases where the source and sink vertices differ from $s$ and $t$, we will specify them explicitly.

\subsection{Feasible Regions}

The shortest path problem over directed graphs is studied only on graphs with no negative-weight cycles, as their presence could make the shortest path have a length of $-\infty$.

Moreover, the standard conversion of an undirected graph to a directed graph (by replacing each undirected edge by two directed edges in opposite directions) could create negative cycles if the undirected graph that we started with had negative weight edges. 

To deal with such situations, we define the notion of a feasible region (similar feasible regions have been used in earlier works~\cite{Carstensen_1983}).
\begin{definition}\label{def:feasibleregion}
    With respect to parametric shortest paths, a feasible region is defined as follows.
    \begin{itemize}
        \item Let $G =(V,E)$ be a directed graph whose edge weights are all linear functions of $\overline{\lambda} = (\lambda_1,\lambda_2\ldots,\lambda_d)$. Then, a region $S\subseteq \mathbb{R}^d$ is called a feasible region for $G$ if for every point $\overline{x}\in\mathbb{R}^d$, the fixed-edge weight graph obtained by substituting $\overline{\lambda}=\overline{x}$ in $G$ contains no negative-weight cycles.
        \item Let $G =(V,E)$ be an undirected graph whose edge weights are all linear functions of $\overline{\lambda} = (\lambda_1,\lambda_2\ldots,\lambda_d)$. Then, a region $S\subseteq \mathbb{R}^d$ is called a feasible region for $G$ if for every point $\overline{x}\in\mathbb{R}^d$, the fixed-edge weight graph obtained by substituting $\overline{\lambda}=\overline{x}$ in $G$ contains no negative-weight edges.
    \end{itemize}
\end{definition}

\subsection{Supporting Results from Combinatorial Geometry}

The analysis of parametric shortest paths often requires results from combinatorial geometry concerning hyperplane arrangements. We state here a classical result that bounds the complexity of such arrangements.

\begin{theorem}[Hyperplane Arrangement Complexity~\cite{Vapnik_Chervonenkis_Ya_2015, Winder_1966}]\label{thm:space_partitioning_theorem}
Consider a $d$-dimensional Euclidean space $\mathbb{R}^d$ that is partitioned by $t$ hyperplanes. Let $r$ denote the number of regions formed by these hyperplanes. Then,
\[
r \leq \sum_{i=0}^{d} \binom{t}{i} \leq t^d + 1.
\]
\end{theorem}

\Cref{thm:space_partitioning_theorem} provides a fundamental upper bound on the complexity of hyperplane arrangements and has direct implications for the number of parametric shortest paths in the setting where the edge weights vary as linear functions of $d$ parameters.

\subsection{Davenport-Schinzel Sequences}\label{ssec:davenport-schinzel_sequences}
The above characterization does not help when the edge weights vary as univariate polynomials of degree $d$. In this setting, like the earlier works, we depend on the combinatorial characterization through Davenport-Schinzel sequences.

\begin{definition}
    Given a finite set of symbols $X$, a sequence $U=(u_1,u_2,\ldots,u_t)$ is a Davenport-Schinzel sequence of order $s$ if it satisfies the following properties.
    \begin{itemize}
        \item $\forall i\in[t]$, $u_i$ is a symbol coming from $X$,
        \item No two consecutive symbols in the sequence $U$ are the same,
        \item If $x_1,x_2\in X$ are distinct symbols, then $U$ doesn't contain a subsequence $(\ldots,x_1,\dots,x_2,\ldots,x_1,\ldots,x_2,\ldots)$ consisting of $s+2$ alternations between $x_1$ and $x_2$.
    \end{itemize}
\end{definition}

In this work, we use Davenport-Schinzel sequences to study the lower envelope formed by a set of univariate polynomials of degree at most $d$. Since any two degree $d$ univariate polynomials can be equal in at most $d$ points, They can alternate at most $d+1$ many times. Therefore, the order of the corresponding Davenport-Schinzel sequence will be $d$. The Davenport-Schinzel sequences have tight bounds when the order of the sequence is constant (see \cite{Davenport_1970}, \cite{Agarwal_Sharir_Shor_1989}, \cite{Nivasch_2010}). However, when the degree $d$ is arbitrary, we do not have good upper bounds on the size of Davenport-Schinzel sequences (of arbitrary order) that we can use. The only known upper bound on the size of Davenport-Schinzel sequences that is applicable here is the trivial $\binom{N}{2}d+1$ (see \cite[p.~3]{Klazar_2002}).

\section{Proofs}\label{sec:generic_upper_bound_technique}


\subsection{Partitioning $\mathbb{R}^d$ using $\mathsf{poly}(n)$ Hyperplanes}
We first present an abstracted result for distinct settings of edge weights that we consider in the paper. We then invoke the necessary space partitioning lemmas in each case and get the final bounds.

\begin{theorem}\label{thm:unifiedbound}
    Let $n, \ell$ be  natural numbers such that $\ell\leq n$. Let $S \subseteq \mathbb{R}^d$. Let $G$ be a single-source and single-sink layered directed acyclic graph with $\ell$ layers and at most $n$ vertices per layer. Let the edge weights be multivariate polynomials, denoted by $f_{e}(\lambda_1, \ldots, \lambda_d)$. Let $N:= n^5$, and let the number of partitions of $S$ through $N$ hyperplanes be $T(N)$. Then,
    \begin{equation*}
        |\PSP(G)| \in T(N)^{O(\log({\ell}))}.
    \end{equation*}
\end{theorem}

\begin{proof}
Without loss of generality, assume that $\ell$ is a power of $2$. Proof of this theorem proceeds via induction on the length of $G$.\\
\textbf{Base case:} Let $\ell = 2$. That is, the source $s$ and the sink $t$ are connected via at most $n$ intermediate nodes $u_1, \ldots, u_n$ as shown in \cref{fig:base_gadget}. For all $i\in[n]$, $i$th path $s\rightarrow u_i \rightarrow t$ is the shortest path for those points $\overline{a}\in S$ that simultaneously satisfy the inequalities $$\mathcal{J}_i = \{f_{(s, u_i)}(\overline{\lambda}) + f_{(u_i, t)}(\overline{\lambda}) \leq f_{(s, u_j)}(\overline{\lambda}) + f_{(u_j, t)}(\overline{\lambda}) \mid j \in [n]\setminus \{i\}\}\,. $$ 

Further, every point in $S$ lies on or in one of the sides of these hyperplanes defined by the following equations.
\begin{align*}
    \mathcal{H}_0 = \{f_{(s, u_i)}(\overline{\lambda}) + f_{(u_i, t)}(\overline{\lambda}) = f_{(s, u_j)}(\overline{\lambda}) + f_{(u_j, t)}(\overline{\lambda}) \mid 1\leq i<j\leq n\}\,.
\end{align*}

In other words, these hyperplanes partition the space $S$ such that each point $\overline{a}\in S$ satisfies at least one of $\mathcal{J}_1, \ldots, \mathcal{J}_n$ and constructively we get that index by computing $\mathrm{argmin}_{i\in [n]}\{f_{(s, u_i)}(\overline{a}) + f_{(u_i, t)}(\overline{a})\}$. It is easy to see that $|\mathcal{H}_0|$ is at most $n^2$. In each of the regions in the partition created by $\mathcal{H}_0$, there is a unique parametric shortest path. So, the number of shortest paths is at most the number of regions formed by $\mathcal{H}_0$.
\begin{equation*}
    |\PSP(G)| \le T(n^2) \le T(n^5).
\end{equation*}

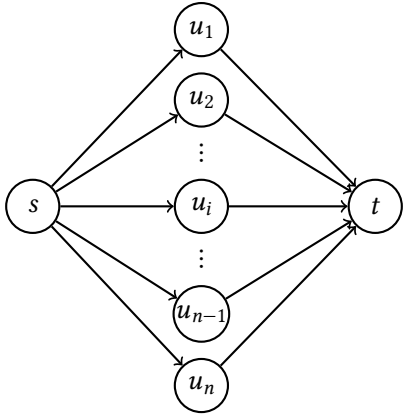
\begin{figure}
\centering
\begin{tikzpicture}[scale=0.9,
    node distance=0.75cm,
    vertex/.style={circle, draw, thick, minimum size=20pt, inner sep=0pt, fill=white} 
]

    \node[vertex] (mid0) {$s$};
    \node[vertex] (mid2) [right=3.8cm of mid0] {$t$}; 

    \node[vertex] (v_i) [right=1.5cm of mid0, yshift=0cm] {$u_i$};
    
    \node (dots_top) [above=0.2cm of v_i, inner sep=0pt] {$\vdots$};
    \node (dots_bottom) [below=0.2cm of v_i, inner sep=0pt] {$\vdots$};

    \node[vertex] (v2) [above=0.2cm of dots_top] {$u_2$};
    \node[vertex] (v1) [above=0.2cm of v2] {$u_1$};
    \node[vertex] (vnminus1) [below=0.2cm of dots_bottom] {$u_{n-1}$};
    \node[vertex] (vn) [below=0.2cm of vnminus1] {$u_n$};

    \path[->, thick]
    (mid0) edge (v1)
    (mid0) edge (v2)
    (mid0) edge (v_i)
    (mid0) edge (vnminus1)
    (mid0) edge (vn)
    
    (v1) edge (mid2)
    (v2) edge (mid2)
    (v_i) edge (mid2)
    (vnminus1) edge (mid2)
    (vn) edge (mid2);

\end{tikzpicture}
    \caption{Graph $G$ of length 2}
    \label{fig:base_gadget}
    \vspace{-0.3cm}
\end{figure}

\noindent\textbf{Inductive hypothesis:} Assume that the statement is true for all lengths $\leq\ell/2$.
\begin{equation*}
    |\PSP(G)| \le T(n^5)^{\log(\frac{\ell}{2})}. \qquad
\end{equation*}
\noindent\textbf{Increment step:} Let length of $G$ be $\ell$. Let $B$ be the set of pairs of vertices in the alternate layers as follows.
$$B = \left\{(u,v)\mid u\in\text{Layer}_{2i-1}, v\in\text{Layer}_{2i+1}~\text{for all}~ 1\leq i\leq \frac{\ell}{2}-1 \right\}\,.$$

For each pair $(u,v) \in B$, let the graph $G^{(u,v)}$ be the subgraph of $G$  induced on the vertices $u,v$ and all the vertices in the layer between $u$ and $v$. This is similar in structure to the graph in the base case (see \cref{fig:base_gadget}). Similar to the base case, we obtain a partition $\pi_{(u,v)}$ of $S$ through the hyperplanes $\mathcal{H}_{(u,v)}$ and within each region thus created, parametric shortest path between $u$ and $v$ is the same for every point.

Let $\mathcal{H} = \cup_{(u,v)\in B}\mathcal{H}_{(u,v)}$. As before, $|\mathcal{H}_{(u,v)}|\leq n^2$ and thus $|\mathcal{H}|\leq n^5$ (since $|B|\leq n^3$). Let $\pi$ be the partition of the space $S$ through the set of hyperplanes $\mathcal{H}$. Let $T = T(n^5)$ and $S_1, \ldots, S_T$ be the regions created by $\pi$.

\begin{observation}\label{obs:regionspaths} For every $j$, all points $\overline{a}$ within $S_j$ are on the same side with respect to all hyperplanes in $\mathcal{H}$. In particular, for every pair $(u,v)\in B$, the shortest path from $u$ to $v$ is the same at all points within $S_j$. For points that lie on the separating hyperplanes, a shortest path is chosen according to a tie-breaking convention (order the vertices of the graph is some arbitrary way, and choose the path that occurs lexicographically earlier).
\end{observation}
Using this observation, we construct $T$ many new graphs $G_j$ corresponding to each region $S_j$, as follows.
\begin{itemize}
    \item Vertex set of $G_j$ consists of all vertices in the odd layers of the graph $G$, and $t$.
    \item For each pair $(u,v)\in B$, connect it with an edge whose edge weight is given by the parametric weight of the unique parametric shortest path between $u$ and $v$ corresponding to the region $S_j$.
    \item Retain all the incoming edges into $t$ along with the original edge weights.
\end{itemize}

We now make the following claim.

\begin{claim}\label{clm:subadditive}
    Let $\PSP_S(G)$ denote the set of $s$ to $t$ paths that show up as shortest paths of $G$ in $S$. Let $\PSP_{S_j}(G_j)$ be defined similarly. Then,
    $$|\PSP_S(G)| \leq \sum_{j=1}^T |\PSP_{S_j}(G_j)|\,.$$
\end{claim}

\begin{proof}[Proof of \Cref{clm:subadditive}:]
    Towards the proof, it is sufficient to establish an injective map from $\PSP_S(G)$ to $\cup_{i=1}^T \PSP_{S_j}(G_j)$, in two steps. 
    \begin{enumerate}
        \item For every path $P\in \PSP_S(G)$ there is a corresponding path $\tilde{P}\in \bigcup_{i=1}^T \PSP_{S_j}(G_j)$.
        \item Two distinct paths $P_1, P_2$ cannot map to the same path in $\bigcup_{i=1}^T \PSP_{S_j}(G_j)$.
    \end{enumerate}
    
    Let $P$ be an arbitrarily chosen path from $\PSP_S(G)$. Let $R_P \subseteq S$ be the region of parameter values over which $P$ is the shortest path from $s$ to $t$ in $G$. Since the regions $S_1, \ldots, S_T$ partition $S$, and $R_P \subseteq S$, we have that $R_P$ must intersect at least one of the regions from $S_1, \ldots, S_T$. Let $R_P \cap S_j \neq \emptyset$ for some $j \in [T]$. 
    A key observation that we make here is that between any two vertices $u \in \text{Layer}_{2i-1}$ and $v \in \text{Layer}_{2i+1}$ that $P$ passes through, the sub-path of $P$ from $u$ to $v$ must be the same as the unique shortest path between $u$ and $v$ in $G$ for every  $\overline{a}\in S_j$. For each pair of consecutive odd layers, we replace each segment of $P$ with the unique shortest path between its endpoints in $S_j$, and concatenating these gives a corresponding path $\tilde{P}\in \PSP_{S_j}(G_j)$.
    
    Let us suppose that two distinct paths $P_1, P_2 \in \PSP_S(G)$ correspond to the same projected path $\tilde{P}\in \PSP_{S_j}(G_j)$. Then $P_1$ and $P_2$ must pass through the same vertices in the odd layers. Now we argue that for every pair $(u,v)$ in the consecutive odd layers, the intermediate vertex between $u$ and $v$ on both the paths must be the same\footnote{Here, we implicitly invoke the tie-breaking convention described in \Cref{obs:regionspaths}.}. If the intermediate nodes for paths $P_1$ and $P_2$ were distinct within a region $S_j$ for a pair $(u,v)\in B$, this contradicts the uniqueness of the shortest path between $u$ and $v$ for that region (where the uniqueness is guaranteed by the definition of $S_j$ and construction of $G_j$).\qedhere
\end{proof}

\cref{clm:subadditive} shows that a path $P\in \PSP_S(G)$ could appear as a projection in various $\PSP_{S_j}(G_j)$ but two distinct paths $P_1, P_2\in\PSP_S(G)$ do not project down to the same path in $\PSP_{S_j}(G_j)$. Graph $G_j$ has length $\ell/2$ (and $\ell/2 + 1$ layers) and from the inductive hypothesis, we have $|\PSP_{S_j}(G_j)| \leq T^{\log(\ell/2)}$ for each $j$. Thus,
\begin{equation*}
    |\PSP_S(G)| \leq T \cdot T^{\log(\ell/2)} = T^{1+\log(\ell/2)} = T^{\log(\ell)} \in T(n^5)^{O(\log(\ell))}.\qedhere
\end{equation*}

\end{proof}

\subsection{Main Result: Directed Acyclic Graphs (\Cref{thm:linear_upper_bound})}\label{ssec:linear_edge_weights}

Here, we prove our main result for directed, acyclic graphs with linear edge weights.

\linearupperbound*

\begin{proof}
    This proof simply invokes \Cref{thm:unifiedbound} for the specific case of linear edge weights. The edge weights are of the form
    \begin{equation*}
        \mathsf{wt}(e) = \sum_{i\in[d]}a_{e,i}\lambda_i + a_{e,d+1}.
    \end{equation*}
    Then, the set $\mathcal{H}$ contains at most $n^5$ hyperplanes and $T(n^5)$ refers to the number of regions formed by the hyperplanes in $\mathcal{H}$. Due to \Cref{thm:space_partitioning_theorem}, we get
    \begin{equation*}
        T(n^5) \le n^{5d} + 1.
    \end{equation*}
    This along with the fact $\ell \le n$ gives us the required upper bound.
    \begin{equation*}
        |\PSP(G)| \le (n^{5d}+1)^{\log(n)} \in n^{O(d\log(n))}.\qedhere
    \end{equation*}
\end{proof}


We now supplement our upper bound with a matching lower bound.

\begin{claim}\label{clm:linearlowerbound}
    For all positive integers $n$, there exists a directed acyclic graph $G^* = (V,E)$ with edge weights of the form $a_{e,1} \lambda_1 + a_{e,2} \lambda_2 + \cdots + a_{e,d} \lambda_d + a_{e,d+1}$ such that $|\PSP(G^*)| \in n^{\Omega(d\log{n})}$.
\end{claim}

\begin{figure}[h]
\vspace{0.85cm}
\centering
    \begin{tikzpicture}[
    node distance=3cm,
    mynode/.style={circle, draw, thick, minimum size=8pt, inner sep=2pt}
]
    \node[mynode] (s1) at (0,0) {$s_1$};
    \node[mynode] (t1) at (3,0) {$t_1$};
    \node[mynode] (s2) at (4.5,0) {$s_2$};
    \node[mynode] (t2) at (7.5,0) {$t_2$};
    \node[mynode] (s3) at (9,0) {$s_3$};
    \node[mynode] (sd) at (12.8,0) {$s_d$};
    \node[mynode] (td) at (15.8,0) {$t_d$};
    
    \tikzset{
        myarrow/.style={thin, densely dashed}
    }
    \draw[myarrow] (s1) to [bend left=50] 
        node[pos=0.5, below, yshift=-11pt, font=\footnotesize] {Edge weights:} (t1);
        
    \draw[thick, ->] (t1) to node[above, font=\footnotesize] {wt. $0$} (s2);
    
    \draw[myarrow] (s1) to [bend right=50] 
        node[pos=0.5, below, yshift=-4pt] {Graph $G_1$} 
        node[pos=0.5, above, yshift=11pt, font=\footnotesize] {$a_{e,1}\lambda_1 + a_{e,d+1}$} (t1);
        
    \draw[myarrow] (s2) to [bend left=50] 
        node[pos=0.5, below, yshift=-11pt, font=\footnotesize] {Edge weights:} (t2);

    \draw[thick, ->] (t2) to node[above, font=\footnotesize] {wt. $0$} (s3);
        
    \draw[myarrow] (s2) to [bend right=50] 
        node[pos=0.5, below, yshift=-4pt] {Graph $G_2$} 
        node[pos=0.5, above, yshift=11pt, font=\footnotesize] {$a_{e,2}\lambda_2 + a_{e,d+1}$} (t2);

    \draw[myarrow] (sd) to [bend left=50] 
        node[pos=0.5, below, yshift=-11pt, font=\footnotesize] {Edge weights:} (td);
        
    \draw[myarrow] (sd) to [bend right=50] 
        node[pos=0.5, below, yshift=-4pt] {Graph $G_d$} 
        node[pos=0.5, above, yshift=11pt, font=\footnotesize] {$a_{e,d}\lambda_d + a_{e,d+1}$} (td);

    \node at (10.9, 0) [font=\LARGE] {$\cdots\cdots\cdots$};
\end{tikzpicture}
\caption{A schematic of the graph $G^*$ used in the proof of \Cref{clm:linearlowerbound}}
\label{fig:Graphstar}
\vspace{0.65cm}
\end{figure}
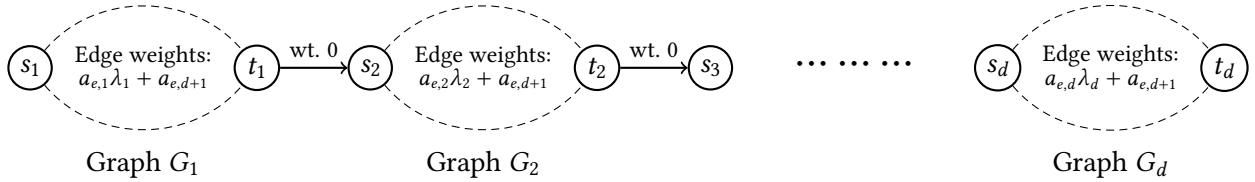

\begin{proof}[Proof Sketch] Fix an $n$ and a $d$. We know there exists an $n$-vertex DAG $G$ with edge weights of the form $a_{e,1} \lambda_1 + a_{e,2}$ whose $\PSP$ is of size $n^{\Omega(\log(n))}$. Such graphs can be found in any of the following papers: \cite{Carstensen_1983,Mulmuley_Shah_2000,Gajjar_Radhakrishnan_2019}. Now, let us describe our construction of the DAG $G^*$ (see \Cref{fig:Graphstar}).

Let $G_1, G_2, \ldots, G_d$ be $d$ identical and disjoint copies of $G$; thus, we have $|\PSP(G_i)|\in n^{\Omega(\log(n))}$ for each $G_i$. The graph $G_1$ has start and end vertices $s_1$ and $t_1$ and edge weights of the form $a_{e,1} \lambda_1 + a_{e,d+1}$, the graph $G_2$ has start and end vertices $s_2$ and $t_2$ and edge weights of the form $a_{e,2} \lambda_2 + a_{e,d+1}$, and so on, the graph $G_d$ has start and end vertices $s_d$ and $t_d$ and edge weights of the form $a_{e,d} \lambda_d + a_{e,d+1}$.

Finally, these graphs are connected with each other in series to obtain $G^*$: there is a $0$-weight directed edge from $t_1$ to $s_2$, there is a $0$-weight directed edge from $t_2$ to $s_3$, and so on, there is a $0$-weight directed edge from $t_{d-1}$ to $s_d$. Hence, the start vertex of the graph $G^*$ is $s_1$ and its end vertex is $t_d$. Since $G_i$ and $G_j$ (for all $i\neq j$) do not have any variables in common, we obtain the following.
\begin{align*}
    |\PSP(G^*)| &= |\PSP(G_1)|\times |\PSP(G_2)|\times \cdots \times |\PSP(G_d)|\\
    &= n^{\Omega(\log(n))} \times n^{\Omega(\log(n))} \times \cdots \times n^{\Omega(\log(n))}\\
    &= n^{\Omega(d\log(n))}.\qedhere
\end{align*}
\end{proof}


\subsection{Directed and Undirected Graphs (\Cref{thm:directedgraphsnonegcyc,thm:undirectedgraphsnonegedg})}\label{subsec:reduction2DAG}

In this section, we show a reduction from undirected and directed graphs to directed acyclic graphs (DAGs). 

Let $G=(V, E)$ be an $n$-vertex graph (either a directed graph with no negative-weight cycles or an undirected graph with non-negative edge weights\footnote{See \Cref{def:feasibleregion} for the meaning of negative-weight cycles and negative-weight edges when the edge weights are not fixed.}) where $V=\{1, \ldots, n\}$ and $s, t \in V$. The weight of each edge $e \in E$ is a function of $d$ real-valued parameters, $\overline{\lambda} = (\lambda_1, \ldots, \lambda_d)$:
$$\mathsf{wt}_e(\overline{\lambda}) = \sum_{j=1}^d a_{e,j} \lambda_j + a_{e,d+1}\,.$$ The total weight of an $s$-$t$ path $P$ is $f_P(\overline{\lambda}) = \sum_{e \in P} \mathsf{wt}_e(\overline{\lambda})$. Since we may assume that we are working in a feasible region (\Cref{def:feasibleregion}) of $G$, for all substitutions of $\overline{\lambda}$, all directed cycles in $G$ are of non-negative weight.

We will now construct a directed acyclic graph $G' = (V', E')$ from $G$ as follows. The new vertex set $V'$ consists of $n$ copies of each vertex spread across $n$ layers. Formally, $V' = \{(v, \ell) \mid v \in V, \ell \in \{0, 1, \ldots, n-1\}\}\,.$ Thus the total number of vertices in $G'$ is $|V'| = n^2$.

A directed edge $(u, i)$ to $(v, i+1)$ is added to the edge set $E'$ in $G'$ if and only if the directed edge $(u, v)$ appears as the $i+1$\textsuperscript{th} edge in any of the simple $s$ to $t$ paths of the original graph $G$ and $i \leq n-1$. 
The weight of an edge in $G'$ is inherited from the corresponding edge in $G$.
    $$\mathsf{wt}'_{((u, i), (v, i+1))}(\overline{\lambda}) = \mathsf{wt}_{(u, v)}(\overline{\lambda})$$

 The new source is $s' = (s, 0)$ and the new set of sinks is $T' = \{(t, \ell) \mid \ell \in \{1, \ldots, n-1\}\}$. Edges only run from layer $i$ to layer $i+1$, and this ensures that the graph is acyclic. A shortest path in $G'$ at a substitution $\overline{a}$ is the shortest among all $(s,0)$ to $(t,i)$ (for all $i$) paths.
 It is important to note that there could be some $s$ to $t$ walks in $G$ appear as simple paths in $G'$. Just from the construction, it is not evident that the parametric shortest paths in $G$ are \emph{preserved} in $G'$ due to addition of some walks of $G$. We will now show that the preservation does happen when we start with a directed graph $G$ with no negative-weight cycles. 

Let $P = (s, u_1, u_2, \ldots, u_{\ell-1}, t)$ be a shortest path of length $\ell$ between $s$ and $t$ in $G$ for a substitution of $\overline{\lambda}$ by $\overline{a}\in \mathbb{R}^d$. From the construction of $G'$, we get that this path $P$ manifests as a path $P' = ((s, 0), (u_1, 1), \ldots, (u_{\ell -1}, \ell-1), (t, \ell))$ in $G'$. We need to show that $P'$ is also a parametric shortest path in $G'$ at $\overline{a}$. Through the inheritance of edge weights in the construction of $G'$, we get that path weights of $P$ and $P'$ are equal. That is, $f_P(\overline{a}) = f_{P'}(\overline{a})$

For the sake of contradiction, let us suppose that there is a path $Q' = ((s,0), (v_1, 1), \ldots, (v_{r-1}, r-1), (t, r))$ that is the shortest path in $G'$ at $\overline{a}$ and $f_P(\overline{a}) > f_{Q'}(\overline{a})$. Let the walk $W = (s, v_1, v_2, \ldots, v_{r-1}, t)$ in $G$ correspond to $Q'$ and thus the weight of the path $Q'$ in $G'$ is equal to the weight of the walk $W$ in $G$. If $W$ were also a simple path and then $f_W(\overline{a}) = f_{Q'}( \overline{a}) < f_{P'}(\overline{a}) = f_P(\overline{a})$. This contradicts the optimality of $P$ at $\overline{a}$.  On the other hand, if $W$ was actually a walk and not a simple path, then elimination of cycles in $W$ (with non-negative weights) would have created a path $P_W$ with a shorter weight or shorter length\footnote{Note that if we started with an undirected graph, this lifting argument would only work if all its edge weights were positive. In particular, if the undirected graph $G$ has a negative-weight edge, it will lead to a negative-weight cycle of length $2$, and a walk that goes back and forth over that $2$-cycle could lead to a non-contradictable path $Q'$ in $G'$ of weight lower than all $s$ to $t$ paths in $G$.}, at $\overline{a}$. This again contradicts the optimality of $P$ at $\overline{a}$.  Using similar arguments, we can handle the case when $f_{P'}(\overline{a}) = f_{Q'}(\overline{a})$.

Combining \Cref{thm:linear_upper_bound} with this reduction, we obtain the following.

\directedgraphsnonegcyc*

\undirectedgraphsnonegcedg*

\subsection{Univariate Polynomial Edge Weights (\Cref{thm:univardegd})}


\univardegd*

\begin{proof}
    This proof is simply a working out of the the proof of \Cref{thm:unifiedbound} for the specific case of univariate, polynomial edge weights. The edge weights are of the form
    \begin{equation*}
        \mathsf{wt}(e) = \sum_{i=0}^{q}a_i\lambda^{i}.
    \end{equation*}
    Then, the set $\mathcal{H}$ contains at most $n^5$ degree $q$, univariate polynomials and $T(n^5)$ refers to the number of line segments the real number line gets split into by polynomials in $\mathcal{H}$. To upper bound $T(n^5)$, we use the trivial upper bound on Davenport-Schinzel sequences (see \Cref{ssec:davenport-schinzel_sequences}).
    \begin{equation*}
        T(n^5) \le \binom{n^5}{2}q + 1.
    \end{equation*}
    This along with the fact that $l\le n$ gives us the required upper bound.
    \begin{equation*}
        |\PSP(G)| \le \left(\binom{n^5}{2}q + 1\right)^{\log(n)} \in (nq)^{O(\log(n))} \in n^{O(\log(n) + \log(q))}.\qedhere
    \end{equation*}
\end{proof}

\section{Shortest Path Oracles}\label{ssec:path_identification}

\subsection{Linear Edge Weights (\Cref{thm:sporaclelin})}

In this subsection, we present a shortest path identification data structure for a parametric graph $G = (V,E)$ with linear edge weights. The data structure does the following: Given $\bar{a}\in\mathbb{R}^d$ as input, it outputs the shortest path from $s$ to $t$ when $\bar{\lambda} =\bar{a}$.

Our construction of this data structure relies on the existing literature on \emph{Point-Location} problem, and we will recall that briefly here.

Let $\mathcal{H}$ denote a set of hyperplanes in the space $\mathbb{R}^d$. The point-location problem is to preprocess $\mathcal{H}$ into a data structure that supports efficient point-location queries. A point-location query inputs a point $\bar{a}\in\mathbb{R}^d$ and asks to identify the cell (formed by hyperplanes in $H$) that contains $\bar{a}$. Toward this, Meiser~\cite{Meiser_1993} introduced a data structure which was later improved by Ezra, Har-Peled, Kaplan \& Sharir~\cite{Ezra_Har-Peled_Kaplan_Sharir_2020}.

\begin{theorem}[Theorem 5.4 \cite{Ezra_Har-Peled_Kaplan_Sharir_2020}]\label{thm:plocation_ds}
    Given a set $\mathcal{H}$ of hyperplanes in the space $\mathbb{R}^d$, there exists a data structure that answers point-location queries in time $O(d^4\log(|\mathcal{H}|))$. The data structure requires $n^{O(d)}$ space and $n^{\tilde{O}(d)}$ preprocessing time.
\end{theorem}

Due to \cref{thm:linear_upper_bound}, $|\PSP(G)| \in n^{O(d\log(n))}$. That is, $\mathbb{R}^d$ is partitioned into $n^{O(d\log(n))}$ many regions with each region corresponding to exactly one path. If we knew all the hyperplanes that partition $\mathbb{R}^d$ into these regions, then we could directly use the data structure from \Cref{thm:plocation_ds}. Since that is not the case at any intermediate step in the proof of \cref{thm:unifiedbound}, we recursively nest the point-location data structure instead.

\sporaclelin*

\begin{proof}
    We begin by observing that every region in any given recursion step in the proof of \cref{thm:unifiedbound} gets partitioned by at most $O(n^5)$ hyperplanes. Given a $\bar{a}\in\mathbb{R}^d$, we can query the point-location data structure to identify the region it belongs to in the first level of recursion. With this knowledge of the region in the first level of recursion, we get $O(n^5)$ more hyperplanes per region in the next level of recursion. In this way, we adaptively query the point-location data structure. Following through the recursion to its last level, we observe that each region in the last level corresponds to exactly one path from $s$ to $t$. This path is the shortest path for all points in that region (including $\bar{a}$). Therefore, we output this path.
    
    To construct the path identification data structure, we nested the point-location data structures in each of the $O(\log(n))$ levels. Note that the number of  instances of the point-location data structure at a depth $i$ of the recursion is at most $O(n^{5id})$.
    \begin{equation*}
        \text{Total number of instances} = \sum_{i=0}^{O(\log{n})} O(n^{5id}) \in  n^{O(d\log(n))}.
    \end{equation*}

    The original point-location data structure has a query time of $O(k^4\log(n))$, requires $n^{O(d)}$ space and uses $n^{\tilde{O}(d)}$ preprocessing time. Thus, the path identification data structure, as constructed above, for a graph with $n$ vertices and edge weights linear in $d$ parameters has the following properties.
    \begin{itemize}
        \item Overall query time is $O(d^4\log^2(n))$,
        \item Total space required is $n^{O(d)}\cdot n^{O(d\log(n))} \in n^{\tilde{O}(d)}$,
        \item Preprocessing time required to construct the data structure is $n^{\tilde{O}(d)}\cdot n^{O(d\log(n))} \in n^{\tilde{O}(d)}$.\qedhere
    \end{itemize}
\end{proof}

\subsection{Univariate Polynomial Edge Weights (\Cref{thm:sporacledeg})}

\sporacledeg*


\begin{proof}

    The data structure we construct is an array in which each element corresponds to a line segment. Every element of the array will be a tuple containing the endpoints of a line segment of $\mathbb{R}$ and the label of the shortest path from $s$ to $t$ corresponding to the line segment.

    To construct this, we first begin with a single element $(-\infty,\infty, G)$ in the array. In the preprocessing step, as shown in the proof of \cref{thm:unifiedbound}, we obtain the subgraphs $\{G_{S_i}\}_{i\in[T]}$ and the endpoints of each $S_i$. We update the array to store this in sorted order. We repeat this for each subgraph until each line segment in the array corresponds to a single path from $s$ to $t$.

    \emph{Query time:}
    When we get a query with the value of $\lambda$, we simply perform a binary search on the array to find the line segment that contains $\lambda$. Then, we output the path that corresponds to the line segment. Since the array can contain at most $n^{O(\log(n)+\log(q))}$ elements, the time taken to respond to a query will be $O(\log^2(n) + \log(n)\log(q))$.

    \emph{Space required and preprocessing time:}
    The number of line segments that are formed in iteration $i$ is given by $(n^{10}q)^i$. Furthermore, there is a $\mathsf{poly}(n)$ overhead for each iteration.

    \begin{equation*}
        \text{Number of line segments} = \sum_{i=0}^{\log{n}} (n^{10}q)^i \in (nq)^{O(\log(n))}.
    \end{equation*}

    This array contains an entry for each parametric shortest path. Therefore, the space required is $n^{O(\log(n)+\log(q))} \in (nq)^{\tilde{O}(1)}$. The total preprocessing time is $(nq)^{\tilde{O}(1)}\mathsf{poly}(n)$.\qedhere
    
\end{proof}

\section{Future Directions}\label{sec:conclusion}

There are several directions of research that can be pursued for parametric shortest paths. Here, we outline three of them.

\begin{enumerate}

    \item[(i)] For all undirected $n$-vertex graphs $G$ with edge weights of the form $a_{e,1} \lambda_1 + a_{e,2} \lambda_2 + \cdots + a_{e,d} \lambda_d + a_{e,d+1}$, is it true that $$|\PSP(G)|\in n^{O(d \log n)}?$$ This seems like the easiest and most immediate open problem to tackle. Since we have already proved the upper bound for DAGs and directed graphs, and we know that such results hold for constant $d$ for undirected graphs, it seems reasonable that this should be true as well.
    
    \item[(ii)] The upper bound for univariate polynomial edge weights (of degree $q$) is either $n^{O(\log(n) + \log(q))}$ or $n^{\log(n)+(\alpha(n)+O(1))^q}$, depending on how $d$ varies with $n$. Improving this upper bound (or proving a matching lower bound) is another interesting avenue. This seems like a tricky problem, but it may require only a few new ideas to go along with the ones we already have.


    \item[(iii)] Finally, the ``holy grail'' of parametric shortest paths would be multivariate polynomials ($d$ variables, degree $q$). The only known lower bound is $n^{\Omega(d\log(n))}$ and no non-trivial upper bound is known. Note that a good upper bound in this setting should be able to recover the upper bound in item (i) (by setting $q=1$), and also recover the upper bound in item (ii) (by setting $d=1$). For example, something like $n^{O(d\log(n) + \log(q))}$ would make sense.
    
    However, simply combining ideas from the $d=1$ and $q=1$ cases seems to be far from enough. The behavior of regions and curves with arbitrary $d$'s and $q$'s varies quite wildly (e.g., self-intersecting curves, disjoint regions for the same weight function), and it seems difficult to capture all of them in a nice combinatorial way. Solving this problem in its full generality may require substantial mathematical insights.
    
\end{enumerate}

\section{Acknowledgments}\label{sec:acknowledgements}

We are deeply grateful to Aryaman Manish Kolhe for many helpful discussions in the earlier stages of this work, especially for studying various toy examples and  plotting their partitions in $\mathbb{R}^2$, which greatly aided our analysis. K.G. thanks Jaikumar Radhakrishnan for hosting him at ICTS, elucidating the proof of Barth, Funke \& Proissl~\cite{Barth_Funke_Proissl_2022}, and the subsequent discussions about it with Prerona Chatterjee, which resulted in a slight improvement to the upper bound~\cite{CGRICTS}. S.C. and K.G. thank Kavitha Telikepalli for suggesting that this work might be useful in shortest path oracles. N.R. thanks Emanuel Juliano for helpful discussions. We also thank Jaikumar Radhakrishnan for carefully verifying the proof of our main result. Finally, we thank the anonymous reviewers of this paper for pointing out some minor errors and for making several helpful suggestions that enhanced its presentation.

N.R. acknowledges partial support from the Dutch Ministry of Education, Culture, and Science through Gravitation project ``Challenges in Cyber Security -- 024.006.037” for this work.

\paragraph{Statement of AI Use:} The authors did not use any AI tools or AI assistants at any stage of this work.

\bibliography{references}

\appendix





\end{document}